\documentclass[letterpaper,10 pt,conference]{ieeeconf}  
\IEEEoverridecommandlockouts                              

\let\labelindent\relax

\usepackage{amsmath, amssymb,bbm,amsthm,bbold}
\usepackage{enumitem}
\usepackage{xparse}
\usepackage{xcolor}
\usepackage{comment}
\usepackage{cite}
\usepackage[english]{babel}
\usepackage{caption,subcaption} 
\usepackage{tikz}
\usepackage{blindtext}
\usepackage{mathrsfs}
\usepackage{graphics} 

\usepackage{hyperref}
\hypersetup{colorlinks=true,  linkcolor=blue,  breaklinks=true,  urlcolor=blue,  citecolor=blue}
\usepackage{soul}
\usepackage{url}

\newtheorem{theorem}{Theorem}
\newtheorem{remark}{Remark}
\newtheorem{example}[theorem]{Example}
\newtheorem{lemma}[theorem]{Lemma}
\newtheorem{definition}[theorem]{Definition}
\newtheorem{proposition}[theorem]{Proposition}

\newtheorem{assumption}{Assumption}
\newtheorem{problem} {Problem}

\title{\LARGE\bf   Data-Driven Distributed Output Synchronization of Heterogeneous Discrete-Time Multi-Agent Systems}

\author{Giulio Fattore  and Maria Elena Valcher
\thanks{G. Fattore and  M.E. Valcher are with the Dipartimento di Ingegneria dell'Informazione,
 Universit\`a di Padova,
    via Gradenigo 6B, 35131 Padova, Italy, e-mail:  \texttt{giulio.fattore@phd.unipd.it, meme@dei.unipd.it}
    }
   }
\date{empty}

\begin{document}
\maketitle
\thispagestyle{empty}
\pagestyle{empty}

\begin{abstract} 
In this paper, we assume that an autonomous exosystem generates a reference output, and we consider
the problem of designing a distributed data-driven
control law 
for a family of discrete-time heterogeneous LTI agents, connected through a directed graph,  in order to synchronize the agents' outputs to the reference one. 
The agents of the network are split into two categories: leaders, with direct access to the exosystem output, and followers, that only  receive information from their neighbors. All agents aim to achieve output synchronization by means of a state feedback that makes use of  their own states 
as well as of
an estimate of the exogenous system state, 
provided  by an internal state observer. 
Such observer has a different structure for leaders and   followers.
 Necessary and sufficient conditions for the existence of a solution are first derived in the model-based set-up and then   in a data-driven context. 
An example illustrates both the implementation procedure and the performance of the proposed  approach.
\end{abstract}
\section{Introduction}\label{sec.intro}

 Output regulation and output synchronization are fundamental control problems, that have attracted the interest of the researchers  since the early seventies.
The solutions to these problems, which are now part of the background of all control engineers, rely on the Internal Model Principle and have their foundations in some milestone papers from giants of the control community
\cite{Davison_TAC76,fw75,fw76,Francis_SIAM77}.

In the early years of this century, there has been a shift from centralized control to distributed control mainly driven by the growing complexity of modern systems, and the need for greater scalability, fault tolerance, and real-time responsiveness.
Distributed control enables decentralized decision-making, improved resilience, and better handling of large-scale, networked environments such as power grids, autonomous vehicles, and industrial automation. This paradigm shift has  triggered a surge of research aiming to provide distributed solutions to classic control problems, including output regulation and synchronization.
Reference \cite{Su_TAC2012} by Su and Huang  pioneered the research on this topic, first exploring how cooperation among agents may allow all of them to synchronize their outputs to a desired  reference trajectory generated by an exosystem.  This first contribution
stimulated a long list of papers
on this topic (see \cite{Cai_AUT2017,Chen_TSMCS2022,Huang_TAC2017,Huo_TCNS2023,Kiumarsi_AUT2017,Modares_AUT2016}, to cite a few). The interested reader is referred to \cite{Chen_AUT2023} for an extensive review of the literature.

In recent years, the control systems community has witnessed another major paradigm shift: from traditional model-based approaches to data-driven methods. This transition is driven by the growing availability of large-scale data, the  advancements in machine learning, the increasingly higher performance of computational methods and lower costs of storage. 
This shift is reflected in a new trend of research, aiming to explore how to 
achieve output synchronization for multi-agent LTI systems when the agent models are not available, but extensive data about their dynamics have been collected in a preliminary offline phase.

To the best of our knowledge,  \cite{Jiao_CDC2021} is the first paper where data-driven output synchronization for (heterogeneous) multi-agent systems (MASs) is investigated. In  \cite{Jiao_CDC2021} agents are described by state-space models that are not   affected by the   exosystem dynamics whose output they want to track. However, a subset of the agents, called ``leaders", has access to the {\em exosystem state} and they share it with their neighbors to ensure that each agent achieves an asymptotic estimate of the exosystem state. The matrices of the state equations are assumed to be unknown, but the matrices involved in the output equations, as well as the exosystem model,  are   known. The design of the state feedback control protocol is based only on data that are collected offline for each agent. 
In \cite{Tian_TII2024}  the   data-driven cooperative output regulation problem is investigated, by imposing a state feedback control strategy that relies on a dead-beat controller. A network of heterogeneous agents is considered, in which   the state dynamics of {\em all} agents are affected by the {\em exosystem state}. All the matrices involved in the agent dynamics are supposed to be unknown. In \cite{Sader_CDC2024} the data-driven output containment control problem for heterogeneous multi-agent systems  is investigated. The   agents dynamics is not affected by the exosystem and all the systems matrices are unknown.\\ 
  It worth remarking that in all these references the exosystem state is accessible to a subset of the agents, nonetheless all agents still 
implement a state observer to estimate the exosystem state that updates based on the state estimates of their neighbors.
Reference \cite{lin_arXiv2025} also deals with
 data driven output regulation, but the authors achieve the goal by introducing multiple copies of the exogeneous systems. Also, the problem solution is obtained via RL techniques.
 Finally, in \cite{Zhou_TCNS2025} a data-driven algorithm is proposed to simultaneously solve an optimal control problem and identify the matrices which describe the system dynamics.

In this paper the output synchronization problem for a discrete-time heterogeneous LTI MAS is investigated. The agents of the network split into two categories: leaders, whose dynamics is affected by  the {\em exosystem output} 
they want to track, and 
followers, that are not affected by it.  In order to synchronize its output with  the exosystem one, each agent implements a state-feedback control strategy, making use of its own state and of its estimate of the exosystem state. However, leaders produce the exosystem state estimate by using a Luenberger observer, while followers achieve this goal by exchanging information about the {\em exosystem output} with their neighbors. 
Necessary and sufficient conditions for the existence of a solution, first by resorting to a model-based approach and then by relying only on data, are derived, together with a (partial) parametrization of the problem solutions. It is shown that, under suitable assumptions on the informativity of the collected data, the data-driven solution is equivalent to the model-based one. An illustrative example concludes the paper showing the soundness of the proposed method.

 The main contributions of this paper compared to the previous literature on this topic are the following ones:
\\
- First of all, we assume what we regard as a more rational problem set-up, by making  the state observers the agents implement
 consistent with their models.
Indeed, if the exosystem directly affects the agent dynamics it makes sense to assume that such agents (the leaders) are aware of this and take advantage of this information to obtain an estimate of the exosystem state in the most efficient way. This means that the state estimate relies on first-hand information received from the exosystem, and not on the state estimates of their neighbors. Meanwhile, agents whose dynamics is not affected by the exosystem (the followers) cannot reasonably have access to direct measurements from the exosystem, and hence the best they can do is to exchange information with their neighbors.
\\
- We assume that leaders have access to the {\em exosystem output}. Since they can provide exact information on such output to their neighbors, it makes sense to exchange the exosystem output estimate
 in order to correct the estimate of the exosystem state (not the exosystem state estimate, nor the agent output).
This makes the overall system dynamics simpler, effective, and easy to design.
\\
- While the previous two comments hold true both for the model-based and the data-driven solutions, focusing on the data-driven solution we can claim that 
our set-up and solutions are more articulated and complete
than those investigated in
\cite{Jiao_CDC2021,Sader_CDC2024}. On the other had, compared to 
\cite{Tian_TII2024} and \cite{lin_arXiv2025} the conditions for solvability we provide are much simpler.\\
- We provide (see Propositions \ref{prop:ParamLeadersCase2} and \ref{prop.follower.stabDD2}) novel conditions to verify when data are informative for stabilization by state feedback \cite{Jiao_CDC2021,vanWaarde_CS_MAG2023} and a parametrization of the solutions.
\\

{\bf Notation.} 
The sets of real numbers and nonnegative integers are denoted by $\mathbb{R}$ and $\mathbb{Z}_+$, respectively. 
Given two integers $h$ and $k$, with $h \le k$, we let $[h,k]$ denote the set $\{h,h+1,\dots,k\}$.
$\mathbb {1}_n$ is the  all-one vector of size $n$. Suffixes will be omitted when the dimensions are irrelevant or can be deduced from the context. 
Given any matrix $Q$,
its {\em Moore-Penrose pseudoinverse} \cite{BenIsraelGreville} is denoted by $Q^{\dagger}$.  If $Q$ is of full 
row rank, then 
$Q^{\dagger} = Q^\top(Q Q^\top)^{-1}$, and 
it is a particular {\em right inverse} of $Q$, by this meaning any (full column rank) matrix $Q^\#$ such that  {\em (s.t.)} $Q Q^\# = I$.

We use ${\rm im}(Q)$ to represent the {\em column space} of $Q$. 
The {\em spectrum} of a square matrix $Q$ is denoted by $\sigma(Q)$ and is the set of all its eigenvalues. 
The {\em Kronecker product} is denoted by $\otimes$.
Given matrices $M_i, i\in [1,p]$,
the  block diagonal matrix whose $i$th diagonal block is the matrix $M_i$ is denoted  either by ${\rm diag} (M_i),i
\in [1,p],$ or by ${\rm  diag}(M_1, M_2, \dots, M_p)$, while given vectors $v_i, i\in [1,p]$, the  column stacking of these vectors is denoted  either by ${\rm col} (v_i), i\in [1,p],$ or by ${\rm  col}(v_1, v_2, \dots, v_p)$.
The $(i,j)$th entry of a matrix $M$ is denoted by $[M]_{i,j}$.

A {\em weighted, directed graph} (digraph) is   a triple ${\mathcal G} = ({\mathcal V}, {\mathcal E}, {\mathcal A})$, where ${\mathcal V} = \{1, \dots, N\}=[1,N]$ is the set of nodes, ${\mathcal E} \subseteq {\mathcal V}\times {\mathcal V}$ is the set of edges, and ${\mathcal A}\in \mathbb{R}^{N\times N}$ is the nonnegative, weighted {\em adjacency matrix} which satisfies $[{\mathcal A}]_{i,j}>  0$ if and only if {\em (iff)} $(j,i)\in {\mathcal E}$. We assume that $[{\mathcal A}]_{i,i} =0$ for every $i\in {\mathcal V}$.
The  {\em in-degree} of the node $i$ is $d_i=\sum_{j=1}^N[{\mathcal A}]_{i,j}$. The {\em in-degree matrix} $\mathcal{D}$ is the diagonal matrix defined as $\mathcal{D}={\rm diag}(d_i), i\in [1,N]$. The {\em Laplacian}  associated with ${\mathcal G} = ({\mathcal V}, {\mathcal E}, {\mathcal A})$ is defined as $\mathcal{L}\doteq \mathcal{D}-\mathcal{A}$. 
A digraph  is {\em connected}  if 
there is a (directed) {\em path} from every node to every other node. 
Given a digraph ${\mathcal G}
 = ({\mathcal V}, {\mathcal E}, {\mathcal A})$, 
a {\em directed spanning tree} is a subgraph of ${\mathcal G}$ that includes all the vertices ${\mathcal V}$ and has a single root node from which all the other nodes can be reached, without forming any cycle.

\section{Error-Feedback Output Synchronization: Problem Statement}\label{sec.Problem_St}
Consider an exosystem described by the equations
\begin{subequations}\label{eq.exo}
\begin{align}
    x_0(t+1)&=Sx_0(t),\label{eq.exo_state}\\
    y_0(t)&=Rx_0(t),\label{eq.exo_out}
\end{align}
\end{subequations}
where $t  \in {\mathbb Z}_+$,  $x_0(t)\in\mathbb{R}^{n_0}$ and $y_0(t)\in\mathbb{R}^{p}$ are the state and the output of the exogenous system, respectively. We make the following standard assumptions on the exosystem.

 \begin{assumption}\label{ass.exo_antistab}
    \cite{Jiao_CDC2021,Kiumarsi_AUT2017}  All the eigenvalues of $S$ are simple and lie on the unit circle.
\end{assumption}
\begin{assumption}\label{ass.exo_obs}
\cite{Jiao_CDC2021,deCarolis_IJRNLC2018}
    The pair $(R,S)$ is   observable.
\end{assumption}
Consider a multi-agent system  consisting of $N$ heterogeneous agents. Without loss of generality, we assume that the first $N_l$ agents (in the following referred to as {\em leaders}) are affected by  the exosystem   output $y_0(t)$, while the remaining $N_f\doteq N-N_l$ agents (the {\em followers}) have not. 
We let ${\mathcal G} = ({\mathcal V}, {\mathcal E}, {\mathcal A})$ denote the directed graph describing the interactions among the $N$ agents, and by ${\mathcal G}_0 = (\{0\} \cup {\mathcal V}, {\mathcal E}_0, {\mathcal A}_0)$ the extended digraph including also the node $0$ corresponding to the exosystem. The set ${\mathcal E}_0$ is obtained by adding to the edges in ${\mathcal E}$ the edges from the exosystem node to the $N_l$ nodes representing the leaders.
We assume that the weights of all such edges are unitary.
Under the previous assumptions, the adjacency matrix ${\mathcal A}_0$ is uniquely identified from ${\mathcal A}$, and the Laplacian matrix associated with it   can be expressed in partitioned form as follows: 
\begin{align} \label{eq.Lap}
    \mathcal{ L}_0=\begin{bmatrix}
        0&\mathbb{0}_{N_l}^\top&\mathbb{0}_{N_f}^\top\\
        {-}\mathbb{1}_{N_l}&\mathcal{L}_{ll}&\mathcal{L}_{lf}\\
        \mathbb{0}_{N_f}&\mathcal{L}_{fl}&\mathcal{L}_{ff}\\
    \end{bmatrix}
\end{align}
where $\mathcal{L}_{ll}\in\mathbb{R}^{N_l\times N_l}$ and $\mathcal{L}_{ff}\in\mathbb{R}^{N_f\times N_f}$. 
\\
We introduce the following assumption on the  digraph $\mathcal{G}_0$.
\begin{assumption}\label{ass.spanning}
\cite{Chen_AUT2023,Jiao_CDC2021,Kiumarsi_AUT2017, Huang_TAC2017}
    The digraph ${\mathcal{G}_0}$  contains a directed spanning tree with root node $0$.
\end{assumption}

The reason why we split the agents in leaders and followers is because, unlike
previous works on this topic, we assume that if the exosystem  output affects the dynamics of an agent,  the agent is aware of receiving direct information from the source it aims at synchronizing with and is
able to transfer the received information to its neighbors. This makes such an agent a leader in the communication process. On the other hand, agents
whose dynamics are not affected by the exogenous system can only rely on the information received from their neighbors, and hence act as followers.

Under the previous assumptions, 
the leader dynamics are
\begin{subequations}\label{eq.leader_i_V2}
\begin{align}
x_i(t+1)&=A_ix_i(t)+B_iu_i(t)+E_iy_0(t),\\
    y_i(t)&=C_ix_i(t)+D_iu_i(t)+F_iy_0(t),
\end{align}
\end{subequations}
for $i\in [1,N_l]$, while the followers  dynamics
are 
\begin{subequations}\label{eq.follower_i}
\begin{align}
    x_i(t+1)&=A_ix_i(t)+B_iu_i(t),\\
    y_i(t)&=C_ix_i(t)+D_iu_i(t),
\end{align}
\end{subequations}
$i \in  [N_l+1,N]$, where $x_i(t)\in\mathbb{R}^{n_i}$, $u_i(t)\in\mathbb{R}^{m_i}$, and $y_i(t)\in\mathbb{R}^{p}$ are the state, input, and output of the $i$th agent, respectively. 
The matrices $A_i, B_i, C_i, D_i, E_i,$ and $F_i$ are real matrices of suitable dimensions.
As it is standard practice in output regulation/synchronization literature, we define the  {\em output tracking error} of the $i$th agent as $e_i(t)\doteq y_i(t)-y_0(t)$. Based on the leaders and followers descriptions \eqref{eq.leader_i_V2} and \eqref{eq.follower_i}, it follows that for $i \in  [1,N_l]$:
\begin{align}\label{eq.leader_out_err_i_V2}
    e_i(t)=C_ix_i(t)+D_iu_i(t)+F_iRx_0(t)-Rx_0(t),
\end{align}
while for $i \in  [N_l+1,N]$ 
\begin{align}\label{eq.follower_out_err_i}
    e_i(t)=C_ix_i(t)+D_iu_i(t)-Rx_0(t).
\end{align}
In this   set-up, every $i$th agent needs to estimate the state of the exosystem, in order to design a state feedback control law that depends on both its state $x_i$ and on its estimate $z_i$   of the exosystem state $x_0$, described as follows:
\begin{equation}
    u_i(t)=K_i(x_i(t)-\Pi_iz_i(t))+\Gamma_iz_i(t),\label{eq.follower_control_u_i}
   \end{equation}
where the matrices $K_i\in {\mathbb R}^{m_i\times n_i}$, $\Pi_i\in {\mathbb R}^{n_i\times n_0}$ and $\Gamma_i\in {\mathbb R}^{m_i\times n_0}$, $i\in [1,N]$,   are design parameters.
Leaders have  direct access to $y_0(t)$ and hence can rely upon a Luenberger observer \cite{SontagBook} to generate the estimate $z_i(t)$ of $x_0(t)$:
\begin{align}\label{eq.leader_control_z_i_V2}
     z_i(t+1)=Sz_i(t)-L(y_0(t)-Rz_i(t)),
\end{align}
$i\in [1, N_l]$, where $L\in {\mathbb R}^{n_0\times p}$, is the observer gain to be designed.
On the other hand,
followers need to exchange information with their neighbors to obtain a reliable asymptotic estimate of $x_0(t)$. 
By extending the philosophy underlying the Luenberger observer, 
followers update their estimate of the exosystem state by collecting from their neighbors the best estimate they can provide of the exosystem output. This means the real exosystem output in the case where the $j$th  neighbor is a leader, and   
$$\hat y_{0,j}(t) \doteq Rz_j(t)$$
in the case where the $j$th neighbor is a follower.
Consequently, the state estimate $z_i(t)$ for the $i$th follower, $i\in [N_l+1,N],$ updates according to:
        \begin{align}
            z_i(t+1)=&
Sz_i(t)
+\frac{1}{1+d_i}H\left[\sum_{j=1}^{N_l}[{\mathcal A}]_{i,j}(y_0(t)-\hat y_{0,i}(t))\right.\nonumber \\
             &\left.+\sum_{j=N_l+1}^N[{\mathcal A}]_{i,j}(\hat y_{0,j}(t)-\hat y_{0,i}(t))\right] \nonumber \\
             =&
Sz_i(t)
+\frac{1}{1+d_i}H\left[\sum_{j=1}^{N_l}[{\mathcal A}]_{i,j}(y_0(t)-Rz_i(t))\right.\nonumber \\
             &\left.+\sum_{j=N_l+1}^N[{\mathcal A}]_{i,j}(Rz_j(t)-Rz_i(t))\right],
        \label{eq.follower_control_z_i_V2}
    \end{align}
 where $[{\mathcal A}]_{i,j}$ is the $(i,j)$th entry of the adjacency matrix ${\mathcal A}$, while 
 $H\in {\mathbb R}^{n_0\times p}$ is a matrix parameter to be designed. 

    \smallskip

  \begin{remark}
To the best of our knowledge, the idea of relying on 
$R z_i(t)$ as an estimate of $y_0(t)$,
 to update the  state estimate dynamics, is original. It provides a simpler algorithm to design the matrices of the observer-based state-feedback controllers compared with the approach that directly employs the agent outputs $y_i(t)$
(see  \cite{Chen_AUT2023,ZhuChen2024}).
\end{remark}

In this scenario, the error-feedback output synchronization problem is stated as follows.

\begin{problem}\label{pb.2} Consider the exosystem \eqref{eq.exo} and the MAS whose leaders are described as in \eqref{eq.leader_i_V2}, $i\in [1,N_l]$, and whose followers are described as in \eqref{eq.follower_i}, $i\in [N_l+1,N]$, and assume that Assumptions \ref{ass.exo_antistab}, \ref{ass.exo_obs},  and
\ref{ass.spanning} 
hold. \\
Design, if possible, matrices $K_i\in {\mathbb R}^{m_i\times n_i}$, $\Pi_i\in {\mathbb R}^{n_i\times n_0}$, $\Gamma_i\in {\mathbb R}^{m_i\times n_0}$, $i\in [1,N]$, $L\in {\mathbb R}^{n_0\times p}$ and $H\in {\mathbb R}^{n_0\times p}$ so that
the overall system, consisting of all the leaders and followers, as well as the state observer equations \eqref{eq.leader_control_z_i_V2} and \eqref{eq.follower_control_z_i_V2}, under the state feedback control law 
\eqref{eq.follower_control_u_i}, $i\in [1,N]$, satisfies the following two conditions:
\begin{enumerate}
    \item if  $x_0(t)$ is identically zero, the system is asymptotically stable;
    \item for all initial conditions $x_0(0)$, $x_i(0)$, and $z_i(0)$,
    \begin{align}
        \lim_{t\to\infty}e_i(t)=\mathbb{0}_p, \quad \forall i\in[1,N].
    \end{align}
\end{enumerate}
\end{problem}
\smallskip

\section{Problem Solution: Model-Based Approach}\label{sec.MB}

For the subsequent analysis it is worth introducing 
(see \cite{Kiumarsi_AUT2017}) 
the $i$th agent {\em state estimation error}
\begin{align}
    \delta_i(t)\doteq z_i(t)-x_0(t),
\end{align}
and
  {\em state tracking error}
\begin{align}\label{eq.state_track_err_i}
    \varepsilon_i(t)\doteq x_i(t)-\Pi_iz_i(t).
\end{align}

\subsection{Leader Nodes Dynamics}

Let us define the global vector corresponding to the leader state dynamics as
$x_l(t) \doteq \mathrm{col}(x_i(t))$,
$i\in [1,N_l]$, and define
$u_l(t)$,  $y_l(t)$,  $z_l(t)$, $e_l(t)$, $\delta_l(t)$, and $\varepsilon_l(t)$ in a similar way.  
Accordingly, we introduce the matrix $A_l \doteq\mathrm{diag}(A_i)$, $i=[1,N_l]$, and  define $B_l$, $C_l$, $D_l$, $E_l$, $F_l$, $K_l$, $\Pi_l$ and $\Gamma_l$  in a similar way. Finally, we define $S_l\doteq(I_{N_l}\otimes S)$ and $R_l\doteq (I_{N_l}\otimes R)$.
We can rewrite the leader dynamics \eqref{eq.leader_i_V2} in compact form as
\begin{subequations}\label{eq.leader_V2}
\begin{align}
x_l(t+1)=&A_lx_l(t)+B_lu_l(t)+E_l(\mathbb{1}_{N_l}\otimes y_0(t)),\\
y_l(t)=&C_lx_l(t)+D_lu_l(t)+F_l(\mathbb{1}_{N_l}\otimes y_0(t)).
\end{align}
\end{subequations}
Also, we can rewrite the state observers equations \eqref{eq.leader_control_z_i_V2} and the state feedback control laws  \eqref{eq.follower_control_u_i} in compact form as
\begin{subequations}\label{eq.leader_control_V2}
    \begin{align}
z_l(t+1)=&(S_l+L_lR_l)z_l(t)-L_l(\mathbb{1}_{N_l}\otimes y_0(t)),\label{eq.leader_control_z_V2}\\
u_l(t)=&K_l(x_l(t)-\Pi_lz_l(t))+\Gamma_lz_l(t).\label{eq.leader_control_u_V2}
    \end{align}
\end{subequations}
Substituting \eqref{eq.leader_control_u_V2} inside \eqref{eq.leader_V2}, we obtain 
\begin{subequations}
    \begin{align}  x_l(t+1)=&(A_l+B_lK_l)x_l(t)+B_l(\Gamma_l-K_l\Pi_l)z_l(t)\nonumber\\
        &+E_l(\mathbb{1}_{N_l}\otimes y_0(t)),\\
        y_l(t)=&(C_l+D_lK_l)x_l(t)+D_l(\Gamma_l-K_l\Pi_l)z_l(t)\nonumber\\
        &+F_l(\mathbb{1}_{N_l}\otimes y_0(t)).
    \end{align}
\end{subequations}
 The dynamics of the state estimation error, of the state tracking error and of the output tracking error become
\begin{subequations}
    \begin{align}
    \delta_l(t+1)=&(S_l+L_lR_l)\delta_l(t),\\
\varepsilon_l(t+1)=&(A_l+B_lK_l)\varepsilon_l(t)\\
   &+[A_l\Pi_l+B_l\Gamma_l-\Pi_l(S_l+L_lR_l)]\delta_l(t)\nonumber\\
   &+(A_l\Pi_l+B_l\Gamma_l+E_lR_l-\Pi_lS_l)(\mathbb{1}_{N_l}\otimes x_0(t)),\nonumber\\
e_l(t)=&(C_l+{D_lK_l})\varepsilon_l(t)+(C\Pi_l+D_l\Gamma_l)\delta_l(t)\\
        &+(C\Pi_l+D_l\Gamma_l+F_lR_l-R_l)(\mathbb{1}_{N_l}\otimes x_0(t)).\nonumber
    \end{align}
\end{subequations}
Note that the leader dynamics is completely decoupled from the follower dynamics. By resorting to a change of basis we can
replace the state vector $[ x_l(t)^\top \ z_l(t)^\top\ ({\mathbb 1}_{N_l}\otimes x_0(t))^\top]^\top$  with $[\delta_l(t)^\top \  \varepsilon_l(t)^\top \ ({\mathbb 1}_{N_l}\otimes x_0(t))^\top]^\top$. Consequently, under Assumptions \ref{ass.exo_antistab}, \ref{ass.exo_obs},  and
\ref{ass.spanning}, we can ensure that the internal dynamics of the leaders are asymptotically stable when $x_0(t)$ is identically zero (see point 1) of Problem \ref{pb.2}) and the error $e_l(t)\to\mathbb{0}$   as $t\to\infty$ 
(see point 2) of Problem \ref{pb.2}) iff
\begin{itemize}
    \item[c1)] There exists a block diagonal matrix $K_l$ s.t. $(A_l+B_lK_l)$ is Schur stable;
        \item[c2)] There exists a block diagonal matrix $L_l$ s.t. $(S_l+L_lR_l)$ is Schur stable;
    \item[c3)] There exist (block diagonal) matrices $\Pi_l$ and  $\Gamma_l$ s.t.
    \begin{subequations}
    \begin{align}
    A_l\Pi_l+B_l\Gamma_l+E_lR_l=&\Pi_l S_l,\\
    C_l\Pi_l+D_l\Gamma_l+F_lR_l=&R_l.
    \end{align}
    \end{subequations}
\end{itemize}
Under condition c3), the dynamics of the state estimation errors, of the state tracking errors and of the output tracking errors become:
\begin{subequations}
\begin{align}
 \delta_l(t+1)=&(S_l+L_lR_l)\delta_l(t),\\
       \varepsilon_l(t+1)=&(A_l+B_lK_l)  \varepsilon_l(t)-(E_l+\Pi_lL_l)R_l\delta_l(t)\\
        e_l(t)=&(C_l+{D_lK_l})\varepsilon_l(t)+(C\Pi_l+D_l\Gamma_l)\delta_l(t).
        \end{align}
\end{subequations}

\subsection{Follower Nodes Dynamics}

Let us define the global vector  corresponding to the follower state dynamics as
$x_f(t)\doteq\mathrm{col}(x_i(t))$, $i\in [N_l+1,N]$, and define  $u_f(t)$,  $y_f(t)$, $z_f(t)$, $e_f(t)$, $ \delta_f(t)$, and $ \varepsilon_f(t)$ in a similar way.
Accordingly, we introduce the matrices $A_f\doteq\mathrm{diag}(A_i)$, $i\in[N_l+1,N]$ and, in a similar way, we define $B_f$, $C_f$, $D_f$, $K_f$, $\Pi_f$ and $ \Gamma_f$. Finally, we define 
$\mathcal{D}_f \doteq \mathrm{diag}(d_i)$, $i\in[N_l+1,N]$, where $d_i$ is the in-degree of the $i$th node of the network, $S_f \doteq (I_{N_f}\otimes S)$ and $R_f \doteq(I_{N_f}\otimes R)$.
We can rewrite the follower dynamics \eqref{eq.follower_i} in  compact form as
\begin{subequations}\label{eq.follower}
\begin{align}
        x_f(t+1)=&A_fx_f(t)+B_fu_f(t),\\
        y_f(t)=&C_fx_f(t)+D_fu_f(t).
\end{align}
\end{subequations}
The state observer equations \eqref{eq.follower_control_z_i_V2} and the state feedback laws \eqref{eq.follower_control_u_i} become
\begin{subequations}\label{eq.follower_control_V2}
    \begin{align}
         z_f(t+1) =&[S_f-(I_{N_f}+\mathcal{D}_f)^{-1}\mathcal{L}_{ff}\otimes HR]z_f(t)\\
    &-[(I_{N_f}+\mathcal{D}_f)^{-1}\mathcal{L}_{fl}\otimes H](\mathbb{1}_{ N_l}\otimes y_0(t)),\nonumber\\
    u_f(t)=&K_f(x_f(t)-\Pi_fz_f(t))+\Gamma_fz_f(t).
    \end{align}
\end{subequations}
The dynamics of the state estimation errors, of the state tracking errors and of the output tracking errors for the followers are:
\begin{subequations}\label{eq.followers.full.nonclean_V2}
    \begin{align}
            \delta_f(t+1)= &[S_f-(I_{N_f}+\mathcal{D}_f)^{-1}\mathcal{L}_{ff}\otimes HR]\delta_f(t),\\
        \varepsilon_f(t+1)=& (A_f+B_fK_f) \varepsilon_f(t)\\
    &+\left[A_f\Pi_f+B_f\Gamma_f-\Pi_fS_f\right.\nonumber\\
    &\left.+\Pi_f(I_{N_f}+\mathcal{D}_f)^{-1}\mathcal{L}_{ff}\otimes HR\right]\delta_f(t)\nonumber\\
    &+(A_f\Pi_f+B_f\Gamma_f-\Pi_fS_f)(\mathbb{1}_{N_f}\otimes x_0(t))\nonumber,\\
e_f(t) =&(C_f+{D_fK_f}) \varepsilon_f(t)\\
&+(C_f \Pi_f+D_f \Gamma_f) \delta_f(t)\nonumber\\
   &+(C_f \Pi_f+D_f \Gamma_f-R_f)
  (\mathbb{1}_{N_f}\otimes x_0(t)).\nonumber    \end{align}
\end{subequations}
Note that the follower dynamics is decoupled, in turn, from the leader dynamics. By resorting to a change of basis we 
replace the state vector $[ x_f(t)^\top \ z_f(t)^\top\ ({\mathbb 1}_{N_f}\otimes x_0(t))^\top]^\top$  with $[\delta_f(t)^\top \  \varepsilon_f(t)^\top \ ({\mathbb 1}_{N_f}\otimes x_0(t))^\top]^\top$.
So, under the proposed assumptions, we can ensure that the internal dynamics of the followers are asymptotically stable (see point 1) of Problem \ref{pb.2}) and $e_f(t) \to {\mathbb 0}$ as $t\to \infty$ (see point 2) of Problem \ref{pb.2})  iff
\begin{itemize}
     \item[c4)] There exists a block diagonal matrix $K_f$ s.t. $(A_f+B_fK_f)$ is Schur stable;
    \item[c5)] There exists a matrix $H$ s.t. $\left[S_f-(I_{N_f}+\mathcal{D}_f)^{-1}\mathcal {L}_{ff}\otimes  HR\right]$ is Schur stable;
    \item[c6)] There exist (block diagonal) matrices $\Pi_f$ and  $ \Gamma_f$ s.t.
    \begin{subequations}
    \begin{align}
    A_f\Pi_f+B_f \Gamma_f=&\Pi_f S_f,\\
    C_f\Pi_f+D_f \Gamma_f=&R_f.
    \end{align}
    \end{subequations}
\end{itemize}
Under  condition c6),  the dynamics of 
\eqref{eq.followers.full.nonclean_V2} become:
\begin{subequations}
\begin{align}
 \delta_f(t+1) =&\left[S_f-(I_{N_f}+\mathcal{D}_f)^{-1}\mathcal {L}_{ff}\otimes  HR\right] \delta_f(t),\\
    \varepsilon_f(t+1)=&(A_f+B_fK_f)  \varepsilon_f(t)\\
  &+\Pi_f\left[(I_{N_f}+\mathcal{D}_f)^{-1}\mathcal {L}_{ff}\otimes  HR\right] \delta_f(t)\nonumber\\
e_f(t) =&(C_f+{D_fK_f}) \varepsilon_f(t)+R_f \delta_f(t).
\end{align}
\end{subequations}

 \subsection{Model-Based Solution}

 As a consequence of the previous analysis, 
the problem solvability reduces to the possibility of satisfying conditions c1)-c6). We have the following result.

\smallskip

\begin{theorem} \label{th:MBsolProb2}
 Consider the exosystem \eqref{eq.exo} and the MAS with leaders  described as in \eqref{eq.leader_i_V2}, $i\in [1,N_l]$, and  followers  described as in \eqref{eq.follower_i}, $i\in [N_l+1,N]$. Under Assumptions \ref{ass.exo_antistab}, \ref{ass.exo_obs},  and 
 \ref{ass.spanning},
 Problem \ref{pb.2} is solvable iff
 \begin{itemize}
   \item[i)] for each $i\in [1,N]$, the pair $(A_i, B_i)$ is stabilizable, 
   \end{itemize}
and  there exist matrices $\Pi_i$, $\Gamma_i$, of suitable dimensions, s.t.
 \begin{itemize}
     \item[ii)] for each $i\in [1,N_l]$, 
     \begin{subequations}\label{eq.leader_cond_MBA}
       \begin{align}
             A_i\Pi_i+B_i\Gamma_i+E_iR&=\Pi_i S,\label{eq.leader_cond_1_MBA}\\
             C_i\Pi_i+D_i\Gamma_i+F_iR&=R,\label{eq.leader_cond_2_MBA}
        \end{align}
         \end{subequations}
         \item[iii)] for each $i\in [N_l+1,N]$,
    \begin{subequations}\label{eq.follower_cond_MBA}
       \begin{align}
             A_i\Pi_i+B_i\Gamma_i&=\Pi_i S,\label{eq.follower_cond_1_MBA}\\
             C_i\Pi_i+D_i\Gamma_i&=R.\label{eq.follower_cond_2_MBA}
        \end{align}
         \end{subequations}
 \end{itemize}
\end{theorem}

\begin{proof}
We first observe that, by Assumption \ref{ass.exo_obs}, the pair $(R,S)$ is observable. This ensures that there exists $L$ such that $S+LR$ is Schur stable, and hence c2) holds for $L_l = I_{N_l}\otimes L$.
On the other hand, 
Assumptions \ref{ass.exo_antistab} and \ref{ass.spanning} ensure that  there exists a matrix  $H$ 
 such that
$S-\lambda HR$ is Schur stable for every $\lambda \in \sigma((I_{N_f} + {\mathcal D}_f )^{-1}{\mathcal L}_{ff})$ (see  Lemmas \ref{lem.eig_Lap} and \ref {lem.stab_2}
in the Appendix).
Corresponding to this matrix  $H$ condition  c5) holds. \\
    We 
observe, now, that there exist matrices $K_i, i\in [1,N]$, such that conditions c1) and c4) hold if and only if each pair $(A_i,B_i)$ is stabilizable, which is exactly condition $i)$.
Finally, 
conditions c3) and c6) are equivalent to $ii)$ and $iii)$.
 
\end{proof}
\smallskip

\section{Data-Driven Approach}\label{sec.DD}

In this section we assume that 
Assumptions \ref{ass.exo_antistab}, \ref{ass.exo_obs}, and \ref{ass.spanning} 
hold, and we introduce a new assumption.

 \begin{assumption}\label{ass.DD}
All the matrices that describe the leaders and followers are unknown, while the matrices $S$ and $R$ which describe the exogenous system dynamics are known.
\end{assumption}

We assume to have collected offline 
  output measurements from the exogenous system, as well as input, state and output measurements from the leader and follower systems on a finite time window  of (sufficiently large) length $T +1$: 
  $y_0^d \doteq \{y_0^d(t)\}_{t=0}^{T-1}$, 
$u_i^d \doteq \{u_i^d(t)\}_{t=0}^{T-1}$, $x_i^d \doteq\{x_i^d(t)\}_{t=0}^{T}$ and $y_i^d \doteq \{y_i^d(t)\}_{t=0}^{T-1}$, $i=[1,N]$.
Accordingly, we set
\begin{subequations}
\begin{align*}
        Y_0^p&\doteq\begin{bmatrix}
        y_0^d(0),\dots,y_0^d(T-1)
    \end{bmatrix}\in \mathbb{R}^{p \times T}\\
    U_i^p&\doteq\begin{bmatrix}
        u_i^d(0),\dots,u_i^d(T-1)
    \end{bmatrix}\in \mathbb{R}^{m_i \times T}\\
      X_i^p&\doteq\begin{bmatrix}
        x_i^d(0),\dots,x_i^d(T-1)
    \end{bmatrix}\in \mathbb{R}^{n_i \times T}\\
      X_i^f&\doteq \begin{bmatrix}
        x_i^d(1),\dots,x_i^d(T)
    \end{bmatrix}\in \mathbb{R}^{n_i \times T}\\
        Y_i^p&\doteq\begin{bmatrix}
        y_i^d(0),\dots,y_i^d(T-1)
    \end{bmatrix}\in \mathbb{R}^{p \times T}    
    \end{align*}
\end{subequations}
In order to solve Problem \ref{pb.2} in a data-driven set-up, we need to understand how conditions $i)$, $ii)$ and $iii)$ in Theorem  \ref{th:MBsolProb2} translate in terms of data matrices. Note that while condition $ii)$ pertains only leaders and condition $iii)$ only followers, condition $i)$ pertains both. However, due to the different state-space models generating the collected data, the conditions we will derive for leaders and followers will be different.

\subsection{Leader Nodes Dynamics}
We start by introducing  the set of all leader systems $(A_i,B_i,E_i,C_i,D_i,F_i)$ that are compatible with the collected data $(Y_0^p,U_i^p,X_i^p,X_i^f,Y_i^p)$, for $i\in [1,N_l]$:
\begin{align*}
\Sigma_i^l\doteq&\left\{(A_i,B_i,E_i,C_i,D_i,F_i)\right.\\&\left. :\ \begin{bmatrix}
        X_i^f\\
        Y_i^p
    \end{bmatrix}=\begin{bmatrix}
        A_i&B_i&E_i\\
        C_i&D_i&F_i
    \end{bmatrix}\begin{bmatrix}
        X_i^p\\U_i^p\\Y_0^p
    \end{bmatrix}\right\}.
    \end{align*}
  We also introduce for every $i\in [1,N_l]$ the set  \begin{align*}
 \Sigma_i^{0,l}\doteq&\left\{(A_i^0,B_i^0,E_i^0,C_i^0,D_i^0,F_i^0)\right.\\&\left. :\ \begin{bmatrix}
        A_i^0&B_i^0&E_i^0\\
        C_i^0&D_i^0&F_i^0
    \end{bmatrix}\begin{bmatrix}
        X_i^p\\
        U_i^p\\
        Y_0^p
    \end{bmatrix}=\begin{bmatrix}
        \mathbb{0}\\
        \mathbb{0}
    \end{bmatrix}\right\}.
\end{align*}
The first aspect we want to address is condition $i)$ for leaders. In a model-based set-up, the stabilizability of each pair $(A_i,B_i)$ is necessary and sufficient for the existence of a matrix $K_i$ such that $A_i+B_iK_i$ is Schur stable. When we rely on data, we are not able to identify the specific sextuple 
$(A_i,B_i,E_i,C_i,D_i,
F_i)$ that has generated the data, and hence in particular the specific pair $(A_i,B_i)$. This means that we need to be able to design from data a matrix $K_i$ that 
makes $A_i+B_iK_i$  Schur stable
for every pair $(A_i,B_i)$ that is compatible with the collected data. For this reason, we  resort to the following definition and characterization.

\begin{definition}
\cite[Definition 3]{Jiao_CDC2021}
The data  $(Y_0^p,U_i^p,X_i^p, X_i^f,$  $Y_i^p)$ are said to be {\em informative for stabilization by state feedback} if there exists a feedback gain $K_i$ s.t. $A_i+B_iK_i$ is Schur stable for all $(A_i,B_i,E_i,C_i,D_i,F_i)\in \Sigma_i^l$.
\end{definition}
\begin{proposition}\cite[Proposition 2]{Jiao_CDC2021} \label{prop.leader.stabDD2}
    The data $(Y_0^p,U_i^p,X_i^p,$ $X_i^f,Y_i^p)$ are informative for stabilization by state feedback iff
    \begin{itemize}
        \item[i)] $X_i^p$ is of full row rank;
        \item[ii)] There exists a right inverse $(X_i^p)^\#$ of $X_i^p$, s.t.
        \begin{itemize}
        \item[iia)] $X_i^f(X_i^p)^\#$ is Schur stable;
        \item[iib)] $Y_0^p(X_i^p)^\#=\mathbb{0}$.
        \end{itemize}
    \end{itemize}
    When so, the stabilizing  feedback gain is $K_i = U_i^p (X_i^p)^\#$.
\end{proposition}

It is worth noticing that, as a consequence of Proposition \ref{prop.leader.stabDD2}, 
if  the data $(Y_0^p,U_i^p,X_i^p,X_i^f,Y_i^p)$ are informative for stabilization by state feedback, then 
$X_i^p$ is of full row rank, and
the matrix
\begin{equation}
    \Psi_i \doteq
    \begin{bmatrix}
        X_i^p \cr Y_0^p \end{bmatrix}
        \in {\mathbb R}^{(n_i+p)\times T}
        \label{eq.defPsii2} \end{equation}
  satisfies ${\rm rank} ( \Psi_i) = {\rm rank} (X_i^p) + {\rm rank} (Y_0^p) = n_i + {\rm rank} (Y_0^p)$. 
  This allows us to introduce a parametrization of all possible 
  right inverses  $(X_i^p)^\#$ of $X_i^p$ that satisfy  $iib)$
  of  Proposition \ref{prop.leader.stabDD2}, as well as a method to check if in this set there exists at least one matrix for which also $iia)$ holds.
  \medskip

 \begin{proposition} \label{prop:ParamLeadersCase2}
Let
$ \Psi_i$ be defined as in \eqref{eq.defPsii2}, with $X_i^p$  of full row rank, and ${\rm rank} ( \Psi_i) = n_i + {\rm rank} (Y_0^p)$, and let
    $\Psi_i^\dag$ be its Moore-Penrose inverse.
    The following facts are equivalent:
\begin{enumerate}
    \item There exists a matrix $(X_i^p)^\#$ such that 
    $\Psi_i (X_i^p)^\#= \begin{bmatrix} I_{n_i}\cr {\mathbb 0}\end{bmatrix}$ and
    $X_i^f(X_i^p)^\#$ is Schur.
        \item The pair $\left(X_i^f \Psi_i^\dag \begin{bmatrix} I_{n_i} \cr {\mathbb 0}\end{bmatrix},  X_i^f\left(I_{T}- \Psi_i^\dag \Psi_i\right)\right)$ is stabilizable.
        \item ${\rm rank}
        \begin{bmatrix}
            X_i^f  \Psi_i^\dag \begin{bmatrix} I_{n_i} \cr {\mathbb 0}\end{bmatrix}
            -\lambda I_{n_i}\!\!\!&\vline &\!\!\! X_i^f\left(I_{T}- \Psi_i^\dag  \Psi_i\right)
        \end{bmatrix}\!\!  
        \!=\!\! n_i,$ $\forall\ \lambda\in\mathbb{C},\ |\lambda|\ge 1$.
        \end{enumerate}
\end{proposition}
\begin{proof}
Let $L_0$ and $R_0$ be a full column rank matrix and a full row rank matrix, respectively, such that
$Y_0^p= L_0 R_0$, so that
$$ \Psi_i = \begin{bmatrix} I_{n_i} & {\mathbb 0}\cr
{\mathbb 0} & L_0\end{bmatrix} 
\begin{bmatrix} X_i^p\cr
R_0\end{bmatrix}.$$
Note that in the previous factorization the matrix on the left is of full column rank, while the matrix on the right is of full row rank and such common rank coincides with
${\rm rank}( \Psi_i)$.
Consequently, by the properties of the Moore-Penrose inverse \cite{Greville_SIAM66}, we can claim that
$$ \Psi_i^\dag =  
\begin{bmatrix} X_i^p\cr
R_0\end{bmatrix}^\dag 
\begin{bmatrix} I_{n_i} & {\mathbb 0}\cr
{\mathbb 0} & L_0\end{bmatrix}^\dag.$$
Now, set
\begin{align} \label{eqparam0}
      \Psi_i^\# \doteq 
      \Psi_i^\dag+\left(I_{T}- \Psi_i^\dag  \Psi_i\right)Q_i,
\end{align}
where $Q_i$ is a free matrix parameter.
\\
We   want to prove that a matrix $(X_i^p)^\#$ satisfies
\begin{equation}
    \label{eq.prm1}
 \Psi_i (X_i^p)^\#= \begin{bmatrix} I_{n_i}\cr {\mathbb 0}\end{bmatrix}
\end{equation}
if and only if it satisfies
\begin{equation}
   \label{eq.prm2} 
(X_i^p)^\# =  \Psi_i^\# \begin{bmatrix} I_{n_i}\cr {\mathbb 0}\end{bmatrix},
\end{equation}
for some $\Psi_i^\#$ is expressed as in \eqref{eqparam0}, i.e., for some  $Q_i$. \\
Clearly, if \eqref{eq.prm2}  holds
for some $Q_i$, then
\begin{align*}
     \Psi_i (X_i^p)^\# 
&=  \Psi_i \left[\Psi_i^\dag+\left(I_{T}- \Psi_i^\dag  \Psi_i\right)Q_i\right]\begin{bmatrix} I_{n_i}\cr {\mathbb 0}\end{bmatrix}\\
&=  \Psi_i  \Psi_i^\dag \begin{bmatrix} I_{n_i}\cr {\mathbb 0}\end{bmatrix}
= \begin{bmatrix} I_{n_i} & {\mathbb 0}\cr
{\mathbb 0} & L_0L_0^\dag \end{bmatrix}\begin{bmatrix} I_{n_i}\cr {\mathbb 0}\end{bmatrix} = \begin{bmatrix} I_{n_i}\cr {\mathbb 0}\end{bmatrix}
\end{align*}
which means that \eqref{eq.prm1} holds.\\
Conversely, 
suppose that \eqref{eq.prm1} holds. 
Since $L_0$ is of full column rank, this means that 
$$\begin{bmatrix} X_i^p\cr
R_0\end{bmatrix} (X_i^p)^\#= \begin{bmatrix} I_{n_i}\cr {\mathbb 0}\end{bmatrix}.$$
Since $\begin{bmatrix} X_i^p\cr
R_0\end{bmatrix}$ is of full row rank, by Lemma \ref{lem:ParamLeadersCase1}, we can claim that
  $(X_i^p)^\#$ can be expressed as
\begin{equation}\label{eq.param1}
(X_i^p)^\#= 
\begin{bmatrix} X_i^p\cr
R_0\end{bmatrix}^\dag \begin{bmatrix} I_{n_i}\cr {\mathbb 0}\end{bmatrix}
+ \left(I_T - \begin{bmatrix} X_i^p\cr
R_0\end{bmatrix}^\dag \begin{bmatrix} X_i^p\cr
R_0\end{bmatrix}\right) Q_i\begin{bmatrix} I_{n_i}\cr {\mathbb 0}\end{bmatrix}
\end{equation}
for some matrix $Q_i$.
It is easy to see that
$$\begin{bmatrix} X_i^p\cr
R_0\end{bmatrix}^\dag \begin{bmatrix} I_{n_i}\cr {\mathbb 0}\end{bmatrix} = 
\begin{bmatrix} X_i^p\cr
R_0\end{bmatrix}^\dag  \begin{bmatrix} I_{n_i} & {\mathbb 0}\cr
{\mathbb 0} & L_0\end{bmatrix}^\dag\begin{bmatrix} I_{n_i}\cr {\mathbb 0}\end{bmatrix}$$
as well as
$$\begin{bmatrix} X_i^p\cr
R_0\end{bmatrix}^\dag \begin{bmatrix} X_i^p\cr
R_0\end{bmatrix}= \begin{bmatrix} X_i^p\cr
R_0\end{bmatrix}^\dag \begin{bmatrix} I_{n_i} & {\mathbb 0}\cr
{\mathbb 0} & L_0\end{bmatrix}^\dag\begin{bmatrix} I_{n_i} & {\mathbb 0}\cr
{\mathbb 0} & L_0\end{bmatrix}\begin{bmatrix} X_i^p\cr
R_0\end{bmatrix}.$$
This shows that \eqref{eq.param1}
coincides with \eqref{eq.prm2}, where 
$ \Psi_i^\#$ is described as in \eqref{eqparam0}
for some $Q_i$.
\\
$1) \Leftrightarrow 2)$ \ From the previous analysis, it follows that condition 1) holds if and only if there exists $Q_i$ such that 
\begin{align*}
    &X_i^f (X_i^p)^\#= X_i^f  \left(\Psi_i^\dag+\left(I_{T}- \Psi_i^\dag  \Psi_i\right)Q_i\right) \begin{bmatrix} I_{n_i} \cr {\mathbb 0}\end{bmatrix} \\
    &= \left(X_i^f  \Psi_i^\dag \begin{bmatrix} I_{n_i} \cr {\mathbb 0}\end{bmatrix}\right)+ \left(X_i^f\left(I_{T}-\Psi_i^\dag \Psi_i\right)\right)Q_i \begin{bmatrix} I_{n_i} \cr {\mathbb 0}\end{bmatrix}
    \end{align*}
    is Schur. 
This is equivalent to saying that the pair 
$  \left(X_i^f    \Psi_i^\dag \begin{bmatrix} I_{n_i} \cr {\mathbb 0}\end{bmatrix},   X_i^f\left(I_{T}-\Psi_i^\dag 
\Psi_i\right)\right)$ is stabilizable (condition 2)).
\\
$2) \Leftrightarrow 3)$ \ Follows from the PBH reachability test \cite{SontagBook}.
 
\end{proof}
\smallskip

We now address condition $ii)$ in Theorem \ref{th:MBsolProb2} from a data-driven perspective.

\begin{proposition} \label{prop.leader.EqnsDD2}
For every $i\in[1,N_l]$, given the data $(Y_0^p,U_i^p,$ $X_i^p,X_i^f,Y_i^p)$, the following conditions are equivalent:
    \begin{itemize}
        \item[1)] There exist matrices $\Pi_i$ and $\Gamma_i$ such that \eqref{eq.leader_cond_MBA} hold for all the sextuples $(A_i,B_i,E_i,C_i,D_i,F_i)\in  \Sigma_i^l$.
       
        \item[2)] There exists a matrix $M_i$ such that   
    \begin{subequations}\label{eq.leader_cond_DD2}
        \begin{align}
            X_i^fM_i &=X_i^pM_iS,\label{eq.leader_cond_1_DD2}\\
            Y_i^pM_i&=R,\label{eq.leader_cond_2_DD2}\\
            Y_0^pM_i&=R.\label{eq.leader_cond_3_DD2}
            \end{align}
             \end{subequations}
        \end{itemize}
\end{proposition}
\begin{proof}  The proof is inspired by the proof in \cite[Lemma 3]{Jiao_CDC2021} and has some similarities with the proof in \cite[Lemma 2]{Tian_TII2024}. \\
$1) \Rightarrow 2)$  
 If we denote by $(A_i,B_i,E_i,C_i,D_i,F_i)$  the ``real" sextuple that generated the data, then clearly $(A_i,B_i,E_i,C_i,D_i,F_i)\in  \Sigma_i^l$ 
and the sextuples $(\hat A_i,\hat B_i,\hat {E}_i,\hat C_i,\hat D_i,\hat {F}_i)$ compatible with the data are those and only those that satisfy
\begin{align*}
 \begin{bmatrix}
     A_i-\hat A_i&B_i-\hat B_i&E_i-\hat {E}_i\\
     C_i-\hat C_i&D_i-\hat D_i&F_i-\hat {F}_i
 \end{bmatrix}\begin{bmatrix}
        X_i^p\\
        U_i^p\\
        Y_0^p
    \end{bmatrix}=\begin{bmatrix}
        \mathbb{0}\\
        \mathbb{0}
    \end{bmatrix},
\end{align*}
which implies that $(A_i- \hat A_i,B_i- \hat B_i,E_i-\hat {E}_i,C_i- \hat C_i,D_i-\hat D_i, F_i- \hat {F}_i)\in \Sigma_i^{0,l}$.
Consequently, 
if the equations \eqref{eq.leader_cond_MBA} hold for all the sextuples   compatible with the data, this means 
that 
for all $(\hat A_i,\hat B_i,\hat {E}_i,\hat C_i,\hat D_i,\hat {F}_i)\in \Sigma_i^l$, it must hold that
\begin{align}
\begin{bmatrix}
     A_i-\hat A_i&B_i-\hat B_i&E_i-\hat E_i\\
     C_i-\hat C_i&D_i-\hat D_i&F_i-\hat F_i
 \end{bmatrix} \begin{bmatrix}
        \Pi_i\\
        \Gamma_i\\
        R  \end{bmatrix}=\begin{bmatrix}
        \mathbb{0}\\
        \mathbb{0}
    \end{bmatrix}.
\end{align}
This implies  that 
\begin{align*}
    {\rm ker_L}\left(\begin{bmatrix}
        X_i^p\\
        U_i^p\\
        Y_0^p
    \end{bmatrix}\right)\subseteq {\rm ker_L}\left(\begin{bmatrix}
        \Pi_i\\
        \Gamma_i\\
        R
    \end{bmatrix}\right)
\end{align*}
where ${\rm ker_L} (Q) \doteq \{ v: v^\top Q = {\mathbb 0}^\top\}.$
This ensures that there exists a matrix $M_i$ such that 
\begin{align*}
    \begin{bmatrix}
        X_i^p\\
        U_i^p\\
        Y_0^p
    \end{bmatrix}M_i=\begin{bmatrix}
        \Pi_i\\
        \Gamma_i\\
        R
    \end{bmatrix}.
\end{align*}
This implies that
$  X_i^pM_i=\Pi_i,
    U_i^pM_i=\Gamma_i, 
    Y_0^pM_i= R,
$
and hence  \eqref{eq.leader_cond_3_DD2} holds.
Moreover, equation \eqref{eq.leader_cond_1_MBA} becomes
\begin{align*}
&A_i \Pi_i+B_i\Gamma_i+E_i R =\Pi_iS,\nonumber\\
\Rightarrow\quad&A_iX_i^pM_i+B_iU_i^pM_i+E_i Y_0^p M_i =X_i^pM_iS,\nonumber\\
\Rightarrow\quad &X_i^fM_i=X_i^pM_iS,
\end{align*}
while  equation \eqref{eq.leader_cond_2_MBA} becomes
\begin{align*}
    &C_i\Pi_i+D_i\Gamma_i+F_iR =R,\nonumber\\
  \Rightarrow\quad   &C_iX_i^pM_i+D_iU_i^pM_i+F_i Y_0^pM_i = Y_0^pM_i,\nonumber\\
 \Rightarrow\quad    &Y_i^pM_i=Y_0^pM_i,
\end{align*}
that are exactly \eqref{eq.leader_cond_1_DD2} and \eqref{eq.leader_cond_2_DD2}.\\
$2) \Rightarrow 1)$ Suppose that there exists a matrix $M_i$ such that \eqref{eq.leader_cond_DD2} hold. 
Set
$\Pi_i\doteq X_i^pM_i,$ and
        $\Gamma_i\doteq U_i^pM_i.$
For every sextuple $(A_i,B_i,E_i,C_i,D_i,F_i)\in  \Sigma_i^l$, condition \eqref{eq.leader_cond_1_DD2}, making use of \eqref{eq.leader_cond_3_DD2}, leads to
\begin{align*}
&(A_iX_i^p+B_iU_i^p+E_iY_0^p)M_i=X_i^pM_iS\nonumber\\
    \Rightarrow\quad    &A_iX_i^pM_i+B_iU_i^pM_i+E_iY_0^pM_i=X_i^pM_iS\nonumber\\
    \Rightarrow\quad &A_i\Pi_i+B_i\Gamma_i+E_i R =\Pi_iS
    \end{align*}
    that is equal to \eqref{eq.leader_cond_1_MBA}.
    \\
    Similarly, starting from \eqref{eq.leader_cond_2_DD2} and using \eqref{eq.leader_cond_3_DD2}, for every sextuple $(A_i,B_i,E_i,C_i,D_i,F_i)\in  \Sigma_i^l$ we get
    \begin{align*}
        &(C_iX_i^p+D_iU_i^p+F_i Y_0^p)M_i=R\nonumber\\
       \Rightarrow\quad &C_iX_i^pM_i+D_iU_i^pM_i+F_iY_0^pM_i=R\nonumber\\
        \Rightarrow\quad &C_i\Pi_i+D_i\Gamma_i+F_i R=R.
    \end{align*}
that is equal to \eqref{eq.leader_cond_2_MBA}.
 
\end{proof}

 \medskip

\subsection{Follower Nodes Dynamics}

For $i\in [N_l+1,N]$, based on equations \eqref{eq.follower_i}, 
we define the set of all   quadruples $(A_i,B_i,C_i,D_i)$, representing followers, that are compatible with the data $(U_i^p,X_i^p,X_i^f,Y_i^p)$ as
\begin{align*}
    \Sigma_i^f\doteq&\left\{(A_i,B_i,C_i,D_i): \begin{bmatrix}
        X_i^f\\
        Y_i^p
    \end{bmatrix}=\begin{bmatrix}
        A_i&B_i\\
        C_i&D_i
    \end{bmatrix}\begin{bmatrix}
        X_i^p\\U_i^p
    \end{bmatrix}\right\}.
    \end{align*}
Also, as in the previous subsection,    we introduce the set
\begin{align*}
\Sigma_i^{0,f}\doteq&\left\{(A_i^0,B_i^0,C_i^0,D_i^0)  \ : \begin{bmatrix}
        A_i^0&B_i^0\\
        C_i^0&D_i^0
    \end{bmatrix}\begin{bmatrix}
        X_i^p\\
        U_i^p
    \end{bmatrix}=\begin{bmatrix}
        \mathbb{0}\\
        \mathbb{0}
    \end{bmatrix}\right\}. 
\end{align*}

We now address the data-driven characterization of condition $i)$ in Theorem \ref{th:MBsolProb2} for the followers.
 \smallskip

\begin{definition}\cite[Definition 12]{vanWaarde_CS_MAG2023}
The data $(U_i^p,X_i^p, X_i^f,$ $Y_i^p)$ are are said to be {\em informative for stabilization by state feedback} if there exists a feedback gain $K_i$ such that $A_i+B_iK_i$ is Schur stable for all $(A_i,B_i,C_i,D_i)\in\Sigma_i^f$.
\end{definition}

The characterization of informativity for stabilization by state feedback  provided in Proposition \ref{prop.follower.stabDD2}, below, is similar to, but much simpler than, the one provided in Proposition \ref{prop:ParamLeadersCase2} and is omitted due to space constraints.

\begin{proposition} \label{prop.follower.stabDD2}
The following facts are equivalent:
\begin{itemize}
    \item[i)] The data $(U_i^p,X_i^p,X_i^f,Y_i^p)$ are informative for stabilization by state feedback.
\item[ii)]  $X_i^p$ is of full row rank, and
    there exists a right inverse $(X_i^p)^\#$ of $X_i^p$ such that $ X_i^f(X_i^p)^\#$ is Schur.
\item[iii)]  $X_i^p$ is of full row rank and
    the pair $\left(
            X_i^f(X_i^p)^\dag,X_i^f(I_T-(X_i^p)^\dag X_i^p)
        \right)$ is stabilizable.
\item[iv)]  $X_i^p$ is of full row rank, and
     ${\rm rank}\left(\begin{bmatrix}
            X_i^f(X_i^p)^\dag-\lambda I_{n_i}\!\!\!&\vline&\!\!\!X_i^f(I_T-(X_i^p)^\dag X_i^p)
        \end{bmatrix}\right)\!\!=\!n_i,$ $ \forall\ \lambda\in\mathbb{C},\ |\lambda|\ge 1$.
    
\end{itemize} 
\end{proposition}

Note that corresponding to every $(X_i^p)^\#$ satisfying $ii)$, we obtain the
 stabilizing  feedback gain  $K_i = U_i^p (X_i^p)^\#$ (see \cite{vanWaarde_CS_MAG2023}). 
The following result is the analogous, for followers,  of Proposition \ref{prop.leader.EqnsDD2} and hence its proof is omitted.

\begin{proposition} \label{prop:ParamFollowersCase2}
For every $i\in[N_l+1,N]$, given the data $(U_i^p,X_i^p,X_i^f,Y_i^p)$,
the following conditions are equivalent: 
    \begin{itemize}
        \item[1)] There exist matrices $\Pi_i$ and $\Gamma_i$ such that \eqref{eq.follower_cond_MBA} hold for all quadruples $(A_i, B_i, C_i, D_i) \in \Sigma_i^f$.
        \item[2)] There exists a matrix $M_i$ such that 
\begin{subequations}\label{eq.follower_cond_DD2}
        \begin{align}
            X_i^fM_i &=X_i^pM_iS,\label{eq.follower_cond_1_DD2}\\
            Y_i^pM_i&=R.\label{eq.follower_cond_2_DD2}
            \end{align}
             \end{subequations}
        \end{itemize}
\end{proposition}

\subsection{Data-driven solution}

By putting together  the  model-solution given in Theorem \ref{th:MBsolProb2}, and the characterizations given for leaders in 
Propositions  \ref{prop:ParamLeadersCase2} and \ref{prop.leader.EqnsDD2}, and for followers in 
Propositions \ref{prop.follower.stabDD2} and \ref{prop:ParamFollowersCase2}, we obtain the complete data-driven solution of Problem \ref{pb.2}.

\begin{theorem} \label{th:DDsolProb2}
 Consider the exosystem \eqref{eq.exo} and the MAS with leaders  described as in \eqref{eq.leader_i_V2}, $i\in [1,N_l]$, and  followers  described as in \eqref{eq.follower_i}, $i\in [N_l+1,N]$. Assume that Assumptions \ref{ass.exo_antistab}, \ref{ass.exo_obs},   \ref{ass.spanning} and \ref{ass.DD} hold.
 Problem \ref{pb.2} is solvable based on the families of collected data $(Y_0^p,U_i^p,X_i^p,X_i^f), i\in [1,N],$ iff  
 \begin{itemize}
  \item[i)] For each $i\in [1,N_l]$, 
  \begin{itemize}
      \item[ia)]
${\rm rank}( \Psi_i)= n_i + {\rm rank}(Y_0^p)$. 
     \item[ib)] 
      The pair $\left(X_i^f\Psi_i^\dag \begin{bmatrix} I_{n_i} \cr {\mathbb 0}\end{bmatrix},  X_i^f\left(I_{T}-\Psi_i^\dag \Psi_i\right)\right)$ is stabilizable.
     \item[ic)] There exists a matrix $M_i$ s.t. equations \eqref{eq.leader_cond_DD2}
     hold.
     \end{itemize}
     \item[ii)] For each $i\in [N_l+1,N]$, 
     \begin{itemize}
     \item[iia)] $X_i^p$ is of full row rank.
     \item[iib)] The pair $\left(
            X_i^f(X_i^p)^\dag,X_i^f(I_{T}-(X_i^p)^\dag X_i^p)
        \right)$ is stabilizable.
     \item[iic)] There exists a matrix $M_i$ s.t. equations \eqref{eq.follower_cond_DD2}
     hold.
     \end{itemize}
 \end{itemize}
\end{theorem}

\begin{example}
Consider an exosystem and a group of 5 agents (2 leaders and 3 followers) connected through the binary (i.e. $[\mathcal{A}]_{ij}$ is either $0$ or $1$) digraph $\mathcal{G}_0$ depicted in Figure \ref{fig.G0} and with dynamics described by the following matrices:

\begin{footnotesize}
\begin{align*}
S&=\begin{bmatrix}
   \sin(0.2)&\cos(0.2)\\
   -\cos(0.2) &\sin(0.2)
\end{bmatrix},\ R=\begin{bmatrix}
    -1&	1
\end{bmatrix},\\
A_1&=\begin{bmatrix}
0&	1&	-1\\
1&	0&	1\\
2&	0&	1
\end{bmatrix},\ B_1=\begin{bmatrix}
1&	0&	0\\
0&	-1&	1\\
1&	0&	2
\end{bmatrix},\ E_1=\begin{bmatrix}
2\\
-1\\
1
\end{bmatrix},\ B_4=\begin{bmatrix}
5\\
5\\
5
\end{bmatrix},\\ 
C_1&=\begin{bmatrix}
    1&	0&	1
\end{bmatrix},\ D_1=\begin{bmatrix}
    1&	2&	3
\end{bmatrix},\ 
C_2=\begin{bmatrix}
    2	&1
\end{bmatrix},\
C_3=\begin{bmatrix}
    1&	1
\end{bmatrix},\\
A_2&=\begin{bmatrix}
0&	2\\
0&	3
\end{bmatrix},\ B_2=\begin{bmatrix}
1\\
1
\end{bmatrix},\ E_2=\begin{bmatrix}
2\\
-1
\end{bmatrix},\
A_3=\begin{bmatrix}
1&	1\\
10&	3
\end{bmatrix},\ B_3=\begin{bmatrix}
3\\
2
\end{bmatrix},\\
A_4&=\begin{bmatrix}
2&	1&	4\\
1	&3&	5\\
0	&0	&4
\end{bmatrix},\ A_5=\begin{bmatrix}
2&	1&	3\\
1&	2&	4\\
0&	0&	4
\end{bmatrix},\ B_5=\begin{bmatrix}
1&	3&	1\\
5&	-3&	6\\
0&	5&	-1
\end{bmatrix},\\
C_4&=\begin{bmatrix}
    1&	2&	3
\end{bmatrix},\ C_5=\begin{bmatrix}
    1&	2&	1
\end{bmatrix},\ D_5=\begin{bmatrix}
3&	6&	-1
\end{bmatrix},\\
F_1&=\begin{bmatrix}
    5
\end{bmatrix},\ D_2=\begin{bmatrix}
3
\end{bmatrix},\ F_2=\begin{bmatrix}
    3
\end{bmatrix},\ D_3=\begin{bmatrix}
6
\end{bmatrix},\ D_4=\begin{bmatrix}
3
\end{bmatrix}.
\end{align*}
\end{footnotesize}

We set $T = 6$. The  inputs $u_i(t), t\in [0,5]$, and the initial states $x_i(0)$ and $x_0(0)$, $i\in [1,5]$,
 have been randomly generated from a standard Gaussian distribution.  We have collected the corresponding
data matrices $(Y_0^p,U_i^p,X_i^p,X_i^f,Y_i^p)$. By relying on the previous analysis, we have obtained the following matrices:

\begin{footnotesize}
    \begin{align*}
    L&=\begin{bmatrix}
    -0.5719\\
   -0.4692
        \end{bmatrix},\
        H=\begin{bmatrix}
    0.1987\\
   -0.9801
        \end{bmatrix},\\
        K_1&=\begin{bmatrix}
    0.7908  &  0.1046  &  0.5590\\
   -0.1677   & 0.2658 &   0.0935\\
   -1.3346 &   0.0327 &  -0.8135
        \end{bmatrix},\\
        K_2&=\begin{bmatrix}
            -0.0001  & -2.8999
        \end{bmatrix},\
        K_3=\begin{bmatrix}
            -1.0303 &  -0.5076
        \end{bmatrix},\\
        K_4&=\begin{bmatrix}
            -2.4279  &  0.7161 &  -0.0281
        \end{bmatrix},\\
        K_5&=\begin{bmatrix}
    3.5372 &   1.1530 &  -1.7219\\
   -0.6923  & -0.2701 &  -0.5844\\
   -3.4587 &  -1.3458 &   0.4763
        \end{bmatrix},\
        \Pi_1=\begin{bmatrix}  
    9.4737 &  -0.7312\\
   -0.8750 &   3.8708\\
    0.0693 &  -3.3919
        \end{bmatrix}\\
         \Pi_2&=\begin{bmatrix}
            0.3327&    3.4521\\
        -1.0572   & 1.1880
        \end{bmatrix},\
         \Pi_3=\begin{bmatrix}
    0.0399 &   0.2869\\
    0.4135 &  -1.0947
        \end{bmatrix},\\
         \Pi_4&=\begin{bmatrix}
    -0.7908 &   0.3916\\
    0.6203  &  1.4994\\
   -2.8368  & -1.0040
        \end{bmatrix},\
         \Pi_5=\begin{bmatrix}
    0.2158 &   0.0351\\
   -0.4961 &   0.1923\\
    0.1232 &  -0.1110
        \end{bmatrix},\\
         \Gamma_1&=\begin{bmatrix}
    5.5430  & -0.1230\\
    4.3996 &  -3.3488\\
  -10.1109 &   1.6856
        \end{bmatrix},\  \Gamma_5=\begin{bmatrix}
    0.0329  &  0.0156\\
   -0.0861  &  0.1020\\
   -0.0710  & -0.0326
   \end{bmatrix},\\
        \Gamma_2&=\begin{bmatrix}
    0.7972 &  -3.3640
        \end{bmatrix},\
        \Gamma_3=\begin{bmatrix}
   -0.2422  &  0.3013
        \end{bmatrix},\\
        \Gamma_4&=\begin{bmatrix}
    2.3536  &  0.2073
        \end{bmatrix}.       
    \end{align*}
\end{footnotesize}

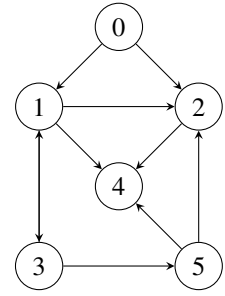
\begin{figure}[ht]
\centering
\begin{tikzpicture}[>=stealth, node distance=1.5cm, every node/.style={draw, circle}]
     \node (0) {0};
    \node (1) [below left of=0] {1};
    \node (2) [below right of=0] {2};
    \node (4) [below right of=1] {4};
    \node (3) [below left of=4] {3};
    \node (5) [below right of=4] {5};
    
    \draw[->] (0) -- (1);
    \draw[->] (0) -- (2);
    \draw[->] (1) -- (2);
    \draw[->] (1) -- (3);
    \draw[->] (1) -- (4);
    \draw[->] (2) -- (4);
    \draw[->] (3) -- (1);
    \draw[->] (3) -- (5);
    \draw[->] (5) -- (2);
    \draw[->] (5) -- (4);
\end{tikzpicture}
\caption{Graph $\mathcal{G}_0$.\label{fig.G0}}
\end{figure}
We have then tested the proposed data driven solution, randomly generating from a standard Gaussian distribution the initial condition $x_i(0)$, $i=[1,5]$, and $x_0(0)$, and the results shown in Figure \ref{fg.plot} highlight the excellent performance of the output synchronization algorithm.

\begin{figure}[ht]
\centering
\subfloat{
  \includegraphics[width=41mm]{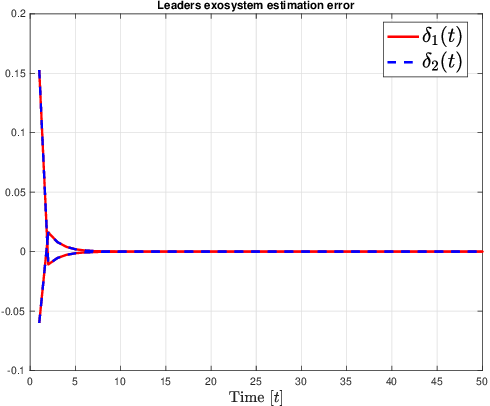}
}
\subfloat{
  \includegraphics[width=41mm]{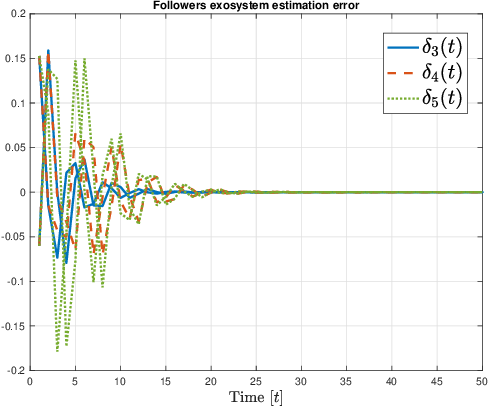}
}
\vspace{3mm}
\subfloat{
  \includegraphics[width=41mm]{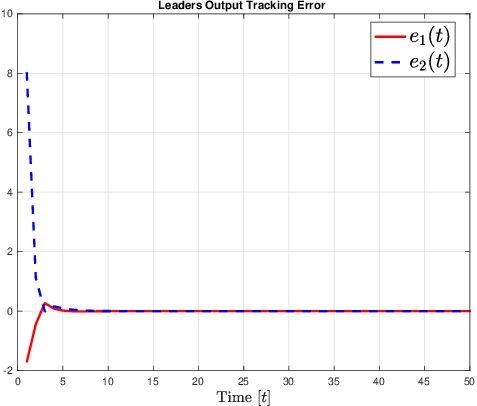}
}
\subfloat{
  \includegraphics[width=41mm]{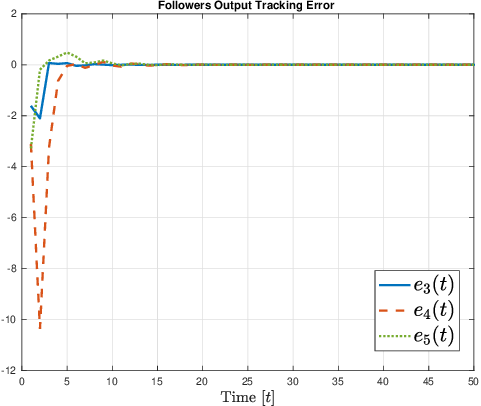}
}
\caption{Plot of the leaders (left) and followers (right) estimation error $\delta_i(t)$ (top) and output tracking error $e_i(t)$ (bottom).}
\label{fg.plot}
\end{figure}
\end{example}

\section{Conclusions}\label{sec.Conclusions}
 
In this paper the output synchronization problem for a discrete-time heterogeneous MAS is solved by relying only on data.  In order to synchronize their outputs with the one of exosystem, each agent implements a state-feedback control strategy making use of its own state and of its estimate of the exosystem state. 
The agents of the network split into two categories, based on whether they have direct access to the exosystem output or not:  leaders and   followers. The former can generate the exosystem state estimate by using a  Luenberger observer, while the others achieve that goal by exchanging information with their neighbors regarding the exosystem output.
Necessary and sufficient conditions for the existence of a solution using only collected data are derived, together with a (partial) parametrization of the problem solutions. It is shown that under suitable assumptions on the informativity of the collected data, the data-driven solution is equivalent to the model-based one. An illustrative example concludes the paper showing the soundness of the proposed method.

\appendix
{
\section{Appendix}
\begin{lemma}\label{lem.eig_Lap}
Under Assumption \ref{ass.spanning}, each eigenvalue  $\lambda$ of $\sigma((I_{N_f}+\mathcal{D}_f)^{-1}\mathcal{L}_{ff})$ 
(see \eqref{eq.Lap}) satisfies 
the following  conditions: a) $\lambda \ne 0$;
b)
$\lambda\in \{z\in {\mathbb C}: |z-1| < 1\}$.
\end{lemma}
\begin{proof}
a) Consider the condensed graph ${\mathcal G}_c$, obtained from ${\mathcal G}_0$ by merging all leader nodes with the node $0$ representing the exosystem. There is an edge from this new node to any of the follower nodes iff there was an edge from one of the leaders to that follower in ${\mathcal G}_0$. Conversely, there is an edge from any follower to the new node iff there was an edge from that follower to one of the leaders in ${\mathcal G}_0$.
It is easy to see that the Laplacian associated with this new digraph is related to 
${\mathcal L}_0$ in \eqref{eq.Lap} as follows:
$${\mathcal L}_c = \begin{bmatrix}
    -N_\ell  + {\mathbb 1}_{N_\ell }^\top {\mathcal L}_{\ell \ell }{\mathbb 1}_{N_\ell } & {\mathbb 1}_{N_\ell }^\top {\mathcal L}_{\ell f}\cr
    {\mathcal L}_{f\ell }{\mathbb 1}_{N_\ell } & {\mathcal L}_{ff}.
\end{bmatrix}$$
 Assumption \ref{ass.spanning}, introduced for ${\mathcal G}_0$, still holds for ${\mathcal G}_c$. But then we can resort to Lemma 1 in \cite{Su_TAC2012} to claim that ${\mathcal L}_{ff}$   is nonsingular square.
As $(I+ {\mathcal D}_f)^{-1}$ is nonsingular, too, then 
a) holds.
 
 b)  By Gershgorin's Circle theorem \cite{HornJohnson}, we can say that
        $\sigma((I_{N_f}+\mathcal{D}_f)^{-1}\mathcal{L}_{ff})$
        is included in $\bigcup_{j\in[N_\ell+1,N]}\left\{z\in\mathbb{C}\ : \ \left|z-\frac{d_j}{1+d_j}\right|\leq\sum_{k=N_\ell+1}^N\frac{[\mathcal{A}]_{jk}}{1+d_j}\right\},$
    but
        $\sum_{k=N_\ell+1}^N [\mathcal{A}]_{jk}\leq d_{j}$,
    and hence for all $j\in[N_\ell+1,N]$ and every $\lambda \in \sigma((I_{N_f}+\mathcal{D}_f)^{-1}\mathcal{L}_{ff})$: 
    \begin{align*}
        \left|\lambda-\frac{d_{j}}{1+d_j}\right|\leq\frac{d_{j}}{1+d_j}
    \end{align*}
    that, together with a), implies b).
\end{proof}

\begin{lemma}\label{lem.stab_2}
Set $\sigma((I_{N_f}+\mathcal{D}_f)^{-1}\mathcal{L}_{ff}) = \{\lambda_1, 
\dots, \lambda_{N_f}\}$. Under Assumptions \ref{ass.exo_antistab}, \ref{ass.exo_obs}
 and \ref{ass.spanning}, there exists a matrix $H$ such that $S-\lambda_k HR$ is Schur stable for all  $k\in [1,N_f]$.
 \end{lemma}
\begin{proof}
First, we want to prove that if $(R,S)$ is observable, then there exists a vector $v\ne \mathbb{0}$ such that $(v^\top R,S)$ is observable. Let $W$ be a nonsingular matrix s.t. $W^{-1}SW={\rm diag}(\mu_1,\dots,\mu_{n_0})$, with $|\mu_j|=1$, $\mu_i\ne\mu_j$ if $i\ne j$. It follows that $(R,S)$ is observable iff $(RW,W^{-1}SW)$ is observable, and by the PBH observability criterion this is the case iff $RW$ does not have  null columns. On the other hand, $(v^\top R,S )$ is observable iff $(v^\top RW,W^{-1}SW)$ is observable, which happens iff $v^\top RW$ does not have  null entries.\\
Let   $\bar w_i$ denote the $i$th column of $RW$, then $v^\top \bar w_i=0$ iff $v\in\left({\rm im}(\bar w_i)\right)^\perp$. Since $\bigcup_{i=1}^{n_0}\left({\rm im}(\bar w_i)\right)^\perp\subsetneqq	\mathbb{R}^p$, it follows that exists $v\in\mathbb{R}^p\setminus\left(\bigcup_{i=1}^{n_0}\left({\rm im}(\bar w_i)\right)^\perp\right)$, namely exists $v\ne\mathbb{0}$ such that $v^\top RW$ does not have null entries, and thus $(v^\top R,S)$ is observable.\\
Now it remains to show that if $(v^\top R,S)$ is observable then there exists $q\in {\mathbb R}^{n_0}$ such that $S-\lambda_k qv^\top R$ is Schur for all $k\in[1,N_f]$, but this is the dual result of \cite[Theorem 3.2]{You_TAC2011}. 
Note that we can apply such result because condition (17) in \cite{You_TAC2011} holds, due to Lemma \ref{lem.eig_Lap}.
So, the result holds for $H= q v^\top.$
\end{proof}

\begin{lemma} \label{lem:ParamLeadersCase1}
Let $A\in {\mathbb R}^{n\times r}$ and  $B\in {\mathbb R}^{k\times r}$, with ${\rm rank}(A)=n$ { and ${\rm rank}(B) = \bar k$}, and suppose that
$$\bar\Psi \doteq \begin{bmatrix}
    A\cr B\end{bmatrix}\in {\mathbb R}^{(n+k)\times r}
$$
has $ {\rm rank}(\bar\Psi)= n+ {\bar k}.$ 
A matrix $C\in {\mathbb R}^{r\times n}$
satisfies 
\begin{equation}\bar\Psi C =  \begin{bmatrix}
    I_n\cr {\mathbb 0}\end{bmatrix}
    \label{pippo1}
    \end{equation}
    if and only if  there exists $Q\in {\mathbb R}^{r \times (n+k)}$ such that
    \begin{equation}
    C = \left(\bar\Psi^\dag  + (I_r -
    \bar\Psi^\dag \bar\Psi) Q\right) \begin{bmatrix}
    I_n\cr {\mathbb 0}\end{bmatrix}.
    \label{pippo2}
    \end{equation}
\end{lemma}

\begin{proof}
Let $L$ and $R$ be a full column rank matrix and a full row rank matrix, respectively, so that
$B= L R$. Then
$$\bar \Psi = \begin{bmatrix} I_{n} & {\mathbb 0}\cr
{\mathbb 0} & L\end{bmatrix} 
\begin{bmatrix} A\cr
R\end{bmatrix}.$$
Note that in the previous factorization the matrix on the left is of full column rank, while the matrix on the right is of full row rank and such common rank coincides with
${\rm rank}( \bar\Psi)$.
Consequently, by the properties of the Moore-Penrose inverse \cite{Greville_SIAM66}, we can claim that
$$ \bar\Psi^\dag =  
\begin{bmatrix} A\cr
R\end{bmatrix}^\dag 
\begin{bmatrix} I_{n} & {\mathbb 0}\cr
{\mathbb 0} & L\end{bmatrix}^\dag.$$
 \\
Clearly, if \eqref{pippo2}  holds
for some $Q$, then
\begin{align*}
     \bar\Psi C 
&=  \bar\Psi \left[\bar\Psi^\dag+\left(I_{r}- \bar\Psi^\dag  
\bar\Psi\right)Q\right]\begin{bmatrix} I_{n}\cr {\mathbb 0}\end{bmatrix}\\
&=  \bar\Psi  \bar\Psi^\dag \begin{bmatrix} I_{n}\cr {\mathbb 0}\end{bmatrix}
= \begin{bmatrix} I_{n} & {\mathbb 0}\cr
{\mathbb 0} & LL^\dag \end{bmatrix}\begin{bmatrix} I_{n}\cr {\mathbb 0}\end{bmatrix} = \begin{bmatrix} I_{n}\cr {\mathbb 0}\end{bmatrix},
\end{align*}
which means that \eqref{pippo1} holds.\\
Conversely, 
suppose that \eqref{pippo1} holds. 
Since $L$ is of full column rank, this { implies} that 
$$\begin{bmatrix} A\cr
R\end{bmatrix} C = \begin{bmatrix} I_{n}\cr {\mathbb 0}\end{bmatrix}.$$
{ Since $R$ is of full row rank, its Moore-Penrose inverse $R^\dag$ is also a right inverse, i.e., $R R^\dag=I_{\bar k}$, and hence
$$\begin{bmatrix}
    A\cr R\end{bmatrix} \begin{bmatrix}
    C & R^\dag\end{bmatrix} = \begin{bmatrix}
    I_n & AR^\dag \cr {\mathbb 0} & I_{\bar k}\end{bmatrix}.$$
Consequently, upon setting $R^\# \doteq R^\dag - CA R^\dag$, we get
    $$ \begin{bmatrix}
    A\cr R\end{bmatrix} \begin{bmatrix}
    C & R^\#\end{bmatrix} = \begin{bmatrix}
    I_n & {\mathbb 0}\cr {\mathbb 0} & I_{\bar k}\end{bmatrix}.$$
    Since
    $\begin{bmatrix}
    C & R^\#\end{bmatrix}$
    is a right inverse of $\begin{bmatrix}
    A\cr R\end{bmatrix}$, there exists some matrix $Q$ such that
    $$\begin{bmatrix}
    C & R^\#\end{bmatrix} = \begin{bmatrix}
    A\cr R\end{bmatrix}^\dag + \left(I_{r} - \begin{bmatrix}
    A\cr R\end{bmatrix}^\dag \begin{bmatrix}
    A\cr R\end{bmatrix}\right) Q.$$}
    This implies that   
\begin{equation}\label{eq.param1_1}
C = 
\left(\begin{bmatrix}
    A\cr R\end{bmatrix}^\dag + \left(I_{r} - \begin{bmatrix}
    A\cr R\end{bmatrix}^\dag \begin{bmatrix}
    A\cr R\end{bmatrix}\right) Q\right)\begin{bmatrix} I_{n}\cr {\mathbb 0}\end{bmatrix},
\end{equation}
for some matrix $Q$.
It is easy to see that
$$\begin{bmatrix} A\cr
R\end{bmatrix}^\dag \begin{bmatrix} I_{n}\cr {\mathbb 0}\end{bmatrix} = 
\begin{bmatrix} A\cr
R\end{bmatrix}^\dag  \begin{bmatrix} I_{n} & {\mathbb 0}\cr
{\mathbb 0} & L\end{bmatrix}^\dag\begin{bmatrix} I_{n}\cr {\mathbb 0}\end{bmatrix} = \bar\Psi^\dag \begin{bmatrix} I_{n}\cr {\mathbb 0}\end{bmatrix}$$
as well as
$$\begin{bmatrix} A\cr
R\end{bmatrix}^\dag \begin{bmatrix} A\cr
R\end{bmatrix}= \begin{bmatrix} A\cr
R\end{bmatrix}^\dag \begin{bmatrix} I_{n} & {\mathbb 0}\cr
{\mathbb 0} & L\end{bmatrix}^\dag\begin{bmatrix} I_{n} & {\mathbb 0}\cr
{\mathbb 0} & L\end{bmatrix}\begin{bmatrix} A\cr
R\end{bmatrix} = \bar\Psi^\dag \bar\Psi.$$
This shows that \eqref{eq.param1_1}
coincides with \eqref{pippo2}, 
and hence $C$ can be written as in \eqref{pippo2} for some $Q$. 
\end{proof}
}

\bibliographystyle{plain}
\bibliography{biblio}

\end{document}